\documentclass[11pt]{article}
\usepackage{fullpage}
\usepackage{subfig}
\usepackage[usenames,dvipsnames]{xcolor}
\usepackage[colorlinks,citecolor=blue,linkcolor=BrickRed]{hyperref}
	\usepackage{latexsym,graphicx,epsfig,color}
\usepackage{amsfonts,amssymb,amsmath,amsthm,amstext}
\usepackage{libertine}
\usepackage[libertine]{newtxmath}
\usepackage{enumitem}
\setitemize{itemsep=2pt,topsep=0pt,parsep=0pt}
\usepackage{url,setspace}
\usepackage{multirow}
\usepackage{rotating}
\usepackage{makeidx}
\usepackage{tikz}
\usepackage{accents}
\usepackage{xspace}
\usepackage{algorithm,algpseudocode}
\usepackage{bm}
\usepackage{changepage} 	
\usepackage{thmtools,thm-restate}

\newcommand{\IGNORE}[1]{}
\allowdisplaybreaks

\usetikzlibrary{decorations.markings}
\usetikzlibrary{arrows}
\tikzstyle{block}=[draw opacity=0.7,line width=1.4cm]
\tikzstyle{graphnode}=[circle, draw, fill=black!20, inner sep=0pt, minimum width=6pt]
\tikzstyle{point}=[circle, draw, fill=black!30, inner sep=0pt, minimum width=1pt]
\tikzstyle{input}=[rectangle, draw, fill=black!75,inner sep=3pt, inner ysep=3pt, minimum width=4pt]
\tikzstyle{unmatched}=[graphnode,fill=black!0]
\tikzstyle{shaded}=[graphnode,fill=black!20]
\tikzstyle{matched}=[graphnode,fill=black!100]  	
\tikzstyle{matching} = [ultra thick]
\tikzset{
    >=stealth',
    pil/.style={
           ->,
           thick,
           shorten <=2pt,
           shorten >=2pt,}
}
\tikzset{->-/.style={decoration={
  markings,
  mark=at position .5 with {\arrow{>}}},postaction={decorate}}}

\makeindex

\newtheorem{theorem}{Theorem}[section]
\newtheorem{claim}[theorem]{Claim}

\newtheorem{lemma}[theorem]{Lemma}
\newtheorem{corollary}[theorem]{Corollary}

\theoremstyle{definition}

\newtheorem{defn}[theorem]{Definition}

\usepackage{thmtools,thm-restate} 

\newcommand{\E}{\mathbb{E}}

\newcommand{\calF}{\mathcal{F}}
\newcommand{\calD}{\mathcal{D}}

\def\bx {{\bf x}}

\def \reals {\mathbb{R}}

\def \one {\chi}

\usepackage{quoting}
\quotingsetup{vskip=0pt}

\newcounter{note}[section]

\newcommand{\x}{\ensuremath{\textbf{x}}}
\newcommand{\Rx}{\ensuremath{ R(\textbf{x}) }}

\newcommand{\univ}{{N}}
\newcommand{\val}{v}
\newcommand{\bval}{\textbf{\val}}
\newcommand{\poly}{\mathrm{poly}}
\newcommand{\M}{\mathcal{M}}
\newcommand{\calm}{\mathcal{M}}
\newcommand{\hcalD}{\hat{\calD}}
\newcommand{\hbval}{\hat{\bval}}
\newcommand{\baseprice}{b}
\newcommand{\arrivaltime}{{T}}
\newcommand{\barrivaltime}{\textbf{T}}
\newcommand{\growingmid}{\mathrel{}\middle|\mathrel{}}
\newcommand{\T}{\ensuremath{\mathcal{T}}\xspace}
\newcommand{\hval}{\hat{\val}}
\newcommand{\Opt}{\mathsf{Opt}}
\newcommand{\Alg}{\mathsf{Alg}}
\newcommand{\Rev}{\mathsf{Revenue}}
\newcommand{\Util}{\mathsf{Utility}}
\newcommand{\Span}{\mathsf{Span}}

\title{Optimal Online Contention Resolution Schemes\\ via Ex-Ante Prophet Inequalities}

	\author{ Euiwoong Lee\thanks{
        (euiwoonl@cs.cmu.edu)
        Courant Institute of Mathematical Sciences,
        New York University.
        }
	\and Sahil Singla\thanks{
        (ssingla@cmu.edu)
        Computer Science Department,
        Carnegie Mellon University.
        }
}

\date{ \today}

\begin{document}
\maketitle

\begin{abstract}
Online contention resolution schemes (OCRSs) were proposed by Feldman, Svensson, and Zenklusen~\cite{FSZ-SODA16} as a generic technique to round a fractional solution in the matroid polytope in an online fashion. It has found  applications  in  several stochastic combinatorial  problems where there is a \emph{commitment} constraint: on seeing the value of a stochastic element, the algorithm has to immediately and irrevocably decide whether to select it while always maintaining an independent set in the matroid. Although OCRSs immediately lead to prophet inequalities, these prophet inequalities are not optimal. Can we instead  use prophet inequalities to design optimal OCRSs?

We design the first optimal $1/2$-OCRS for matroids by reducing the problem to designing a matroid prophet inequality where we compare  to the stronger benchmark of an ex-ante relaxation. We also introduce and design optimal $(1-1/e)$-random order CRSs for matroids, which are similar to OCRSs but   the arrival is chosen uniformly at random.

 \end{abstract}



\section{Introduction}

Given a combinatorial optimization problem, a common algorithmic approach  is to first solve a convex relaxation of the problem  and to then round the obtained {fractional} solution $\bx$ into a \emph{feasible integral} solution while (approximately) preserving the objective. Contention resolution schemes (CRSs), introduced in~\cite{CVZ-SICOMP14}, is a way to perform this rounding given a fractional solution $\x \in \reals_{\geq 0}^n$. For $c>0$, intuitively a $c$-CRS is a rounding algorithm that guarantees every element $i$  is selected into the final feasible  solution w.p. at least $c \cdot x_i$. For a \emph{maximization} problem with a linear objective, by linearity of expectation such a $c$-CRS directly implies a $c$-approximation algorithm.

In a recent work, Feldman et al.~\cite{FSZ-SODA16} introduced an Online CRS (OCRS), which is a CRS with an additional property that it performs the rounding in an ``online fashion''. This property is crucial for the \emph{prophet inequality} problem  (or any stochastic combinatorial problem with a \emph{commitment} constraint; see~\S\ref{sec:related}).

\begin{defn}[Prophet inequality] Suppose each element $i \in \univ$ takes a value $\val_i \in \reals_{\geq 0}$ independently from some known distribution $\calD_i$. These values are presented  one-by-one  to an online algorithm in an adversarial order. 
Given a packing feasibility constraint $\calF \subseteq 2^{\univ}$, the problem is to  {immediately} and {irrevocably} decide whether to select the next element $i$, while always maintaining a feasible solution and maximizing the sum of the selected values.
\end{defn}
 A $c$-\emph{approximation prophet inequality} for $0\leq c\leq 1$ means there exists an online algorithm with expected value at least $c$ times the expected value of an offline algorithm that knows all values from the beginning. 
As shown in~\cite{FSZ-SODA16}, a $c$-OCRS   immediately  implies a $c$-approximation prophet inequality. Some other applications are  oblivious posted pricing mechanisms and stochastic probing.

Although powerful, the above approach of using OCRSs to design prophet inequalities  does not  give us \emph{optimal} prophet inequalities. For example, while we know a $1/2$-approximation prophet inequality over matroids~\cite{KW-STOC12}, we only know a $1/4$-OCRS over matroids~\cite{FSZ-SODA16}. This indicates that the currently known OCRSs may not be optimal.
Can we design better OCRSs?
The main contribution of this work is to design an \emph{optimal} OCRS over matroid constraints using the following idea:
\begin{quoting}
\emph{Not only can we design prophet inequalities from OCRSs, we can also design OCRSs from prophet inequalities.}
\end{quoting}
More specifically, our OCRS is based on  an \emph{ex-ante} prophet inequality: we compare the online algorithm to the  stronger benchmark of a convex relaxation. We   modify existing prophet inequalities
 to obtain  ex-ante prophet inequalities while \emph{preserving} the {approximation} factors. 
As a corollary, this gives  the first {optimal} $1/2$-OCRS over matroids.

Since for many applications  the arrival order is not chosen by an adversary, some recent works have also studied  prophet secretary inequalities where the  arrival order is chosen uniformly at random~\cite{EHLM-SIDMA17,EHKS-SODA18,ACK-ArXiv17}. Motivated by these works,  we  introduce  \emph{random order contention resolution schemes} (RCRS), which is an OCRS for uniformly random arrival\footnote{A parallel independent work has also introduced RCRS~\cite{AW-ArXiv18}; however,  their technical results are very different.}. Again by designing the corresponding random order ex-ante prophet  inequalities, we obtain optimal $(1-1/e)$-RCRS over matroids.

In \S\ref{sec:modelResults} we formally define  an OCRS/RCRS and an ex-ante prophet inequality. In \S\ref{sec:resultsAndTechniques} we describe our results  and proof techniques.
See  \S\ref{sec:related} for further related work.

\subsection{Model} \label{sec:modelResults}

CRSs are a powerful tool for  offline and stochastic optimization problems~\cite{CVZ-SICOMP14,GN-IPCO13}. For a given $\x \in [0,1]^\univ$, let $\Rx$ denote a random set containing each element $i\in \univ$   independently w.p. $x_i$. We say an element $i $ is \emph{active} if it belongs to $\Rx$.
\begin{defn}[Contention resolution scheme]
Given a finite ground set $\univ$ with $n = |\univ|$ and a packing ({downward-closed}) family of feasible  subsets $\calF \subseteq 2^\univ$, let $P_{\calF} \subseteq [0, 1]^\univ$ be the convex hull of all characteristic vectors of feasible sets. For a given $\x \in P_{\calF}$,  a \emph{$c$-selectable} CRS (or simply,  $c$-CRS) 
is a (randomized) mapping $\pi : 2^\univ \to 2^\univ$ satisfying the following three properties: 
\begin{enumerate}[itemsep=0ex,label=(\roman*),topsep=0pt,parsep=0pt]
\item $\pi(S) \subseteq S$ for all $S \subseteq \univ$. 	\label{eq:CRSproperty1}
\item $\pi(S) \in \calF$ for all $S \subseteq \univ$.	\label{eq:CRSproperty2}
\item $\Pr_{\Rx, \pi} [i \in \pi(\Rx)] \geq c \cdot x_i$ for all $i \in \univ$. \label{eq:CRSproperty3}
\end{enumerate}
\end{defn}
Notice, if $f$ is a monotone linear function then $\E[f(\pi(\Rx))] \geq c \cdot \E[f(\Rx)]$. 
By constructing CRSs for various constraint families of $\calF$, Chekuri et al.~\cite{CVZ-SICOMP14} give improved approximation algorithms for linear and submodular maximization problems under knapsack, matroid, matchoid constraints, and their intersections\footnote{Some  ``greedy'' properties are also required from the CRS for the guarantees to  hold for a submodular function $f$~\cite{CVZ-SICOMP14}.}. 

In the above applications to offline optimization problems, the algorithm first flips all the random coins to sample $\Rx$, and then  obtains $\pi(\Rx) \subseteq \Rx$. 
For various online problems such as the prophet inequality,  
this randomness is an inherent part of the problem.
 Feldman et al.~\cite{FSZ-SODA16} therefore introduce an OCRS  where the random set $\Rx$ is sampled in the same manner, but whether $i \in \Rx$ (or not) is only revealed one-by-one to the algorithm in an adversarial order\footnote{For adversarial arrival order, we assume that this order   is known to the OCRS algorithm in advance. This offline  adversary is weaker than the almighty adversary  considered in~\cite{FSZ-SODA16}, but is common in the prophet inequality literature~\cite{Rubinstein-STOC16,RS-SODA17}. We need this assumption in \S\ref{sec:ExAnteToOCRS} to  define our exponential sized linear program.}. 
 After each revelation (arrival), the OCRS has to irrevocably decide whether to include $i \in \Rx$  into $\pi(\Rx)$ (if possible). 
 A $c$-selectable OCRS (or simply, $c$-OCRS) is an OCRS   satisfying the above properties~\ref{eq:CRSproperty1} to \ref{eq:CRSproperty3} of a $c$-CRS. 

In this work, we also study RCRS which is an OCRS with the arrival order chosen  uniformly at random.   A $c$-selectable RCRS (or simply, $c$-RCRS) is an RCRS   satisfying the above properties~\ref{eq:CRSproperty1} to \ref{eq:CRSproperty3} of a $c$-CRS, where in Property~\ref{eq:CRSproperty3} we also take expectation over the arrival order.

While prophet inequalities have been designed using OCRSs, our main result in this paper is to  show a deeper {\em reverse} connection between OCRSs and prophet inequalities.
We first define an ex-ante prophet inequality. Given a prophet inequality problem instance with  packing constraints $\calF$ and r.v.s $\val_i \sim \calD_i$ for $i\in \univ$,   the following \emph{ex-ante relaxation}  gives an upper bound on the expected offline optimum:
\begin{align}\label{eq:exanteUpper}
	\max_{\x} ~\sum_{i} x_i \cdot \E_{\val_i\sim \calD_i}[\val_i \mid \text{$\val_i$ takes value in its top $x_i$ quantile}] \quad \text{s.t.} \quad \x \in P_{\calF}.
\end{align}
To prove that \eqref{eq:exanteUpper} is an upper bound,  we interpret   $x_i$ as the probability that $i$ is in the offline optimum. It is also known that \eqref{eq:exanteUpper}  is a convex program and can be solved efficiently; see~\cite{FSZ-SODA16} for more details.
\begin{defn}[Ex-ante prophet inequality] For $0\leq c\leq 1$, a $c$-approximation ex-ante prophet inequality for packing constraints $\calF$ is a prophet inequality algorithm with expected value at least $c$ times \eqref{eq:exanteUpper}.
\end{defn}

Before describing our results, to build some intuition for the above definitions we  discuss the special case of a rank~$1$ matroid, i.e., where we can only select   one of the $n$ elements.

\paragraph*{Example: Rank~$1$ matroid}
%
For simplicity, in this section we assume that all random variables are Bernoulli, i.e., $\val_i$ takes value $y_i$ independently w.p. $p_i$, and is $0$ otherwise. We first show why a $c$-OCRS implies a $c$-approximation prophet inequality for rank~$1$ matroids.

Consider the optimum solution $\x$  to the ex-ante relaxation~\eqref{eq:exanteUpper} for the above Bernoulli instance.  Its objective value is $\sum_i x_i y_i$ where $\x$   satisfies   $\sum_i x_i \leq 1$. Moreover,  $x_i \leq p_i$ for all $i$ because selecting $i$ beyond $p_i$ does not increase~\eqref{eq:exanteUpper}. To see why  \eqref{eq:exanteUpper} gives an upper bound on the expected offline  maximum, observe that if we interpret $x_i$ as the probability that $\val_i$ is the offline maximum, this  gives a feasible solution to $\sum_i x_i \leq 1$ and with value at most  $\sum_i x_i y_i$.
Thus, to prove a $c$-approximation prophet inequality, it suffices to  design an online algorithm with value at least $c \cdot \sum_i x_i y_i$. 
Consider an algorithm that runs a $c$-OCRS on $\x$, where $i$ is considered active independently w.p. $x_i/p_i$ whenever $\val_i$ takes value $y_i$. This ensures element $i$ is active w.p. exactly $x_i$. Since a $c$-OCRS guarantees each element is selected w.p. $\geq c$ when it is active,  by linearity of expectation such an algorithm has expected value at least $ c \cdot \sum_i x_i y_i$. 

We now discuss a simple $1/4$-OCRS for a rank~$1$  matroid. 
Given $\x$ satisfying $\sum_i x_i \leq 1$, consider an algorithm that  ignores each element $i$ independently w.p. $1/2$, and otherwise  selects $i$ only if it is active. Since this algorithm selects any element $i$ w.p. at most $x_i/2$ (when $i$ is not ignored and is active), by Markov's inequality the algorithm selects no element till the end w.p. at least $1 - \sum_i x_i/2 \geq 1/2$. Hence the algorithm reaches  each element $i$ w.p. at least $1/2$ without selecting any of the previous elements. Moreover,  it does not ignore $i$ w.p. $1/2$, which implies it considers each element w.p. at least $1/4$. The OCRS due to Feldman et al.~\cite{FSZ-SODA16} can be thought of generalizing this approach to a general matroid.

An interesting result of Alaei~\cite{Alaei-SICOMP14} shows  that the above $1/4$-OCRS can be improved to  a $1/2$-OCRS over a rank $1$ matroid by  ``greedily'' maximizing the probability of ignoring the next element $i$, but considering $i$ w.p. $1/2$ on average.
In \S\ref{sec:OCRSProofsSingleItem} we present Alaei's proof for completeness. In \S\ref{sec:RCRSProofsSingleItem}, we also show how to obtain a simple $(1-1/e)$-RCRS for a rank~$1$ matroid.  This raises the question whether one can obtain a $1/2$-OCRS and a $(1-1/e)$-RCRS for general matroids.


\subsection{Results and Techniques} \label{sec:resultsAndTechniques}

Our first theorem gives an approximation factor preserving reduction from OCRSs to  ex-ante prophet inequalities.

\begin{theorem}
For $0\leq c \leq 1$, a $c$-approximation ex-ante prophet inequality for adversarial (random) arrival order over a packing constraint $\calF$ implies  a $c$-OCRS ($c$-RCRS)  over $\calF$.
\label{thm:reduction}
\end{theorem}

We complement the above theorem by designing ex-ante prophet inequalities over matroids.

\begin{theorem} \label{thm:ExAnteProphetIneq}
For matroids, there exists a  $1/2$-approximation ex-ante prophet inequality   for  adversarial arrival order and a $(1-1/e)$-approximation ex-ante prophet inequality   for uniformly random arrival order.
\end{theorem}

As a corollary, the above two theorems give  optimal OCRS and RCRS over matroids.
This generalizes the rank~$1$ results discussed in the previous section to general matroids; although the proof techniques are very different.
 
\begin{corollary} \label{cor:OCRSMain}
For matroids, there exists a $1/2$-OCRS   and a  $(1-1/e)$-RCRS.
\end{corollary}
Our $1/2$-OCRS above assumes that the arrival order is known to the algorithm. It is an interesting open question to find  a $1/2$-OCRS  for an almighty/online adversary as  in~\cite{FSZ-SODA16}.

We first prove that both the  factors $1/2$ and $(1-1/e)$ in Corollary~\ref{cor:OCRSMain} are \emph{optimal}.


\paragraph*{Optimality of $1/2$-OCRS and $(1-1/e)$-RCRS}  \label{sec:OptimalityOfOCRS}
We  argue that the  factors  $1/2$ and $(1-1/e)$  in Corollary~\ref{cor:OCRSMain} are optimal even in the special case of a rank~$1$  matroid. For adversarial arrival, consider just two elements, i.e., $n=2$, with $x_1 = 1-\epsilon$ and $x_2=\epsilon$ for some $\epsilon \rightarrow 0$. Since the OCRS algorithm has to select the first element at least $1/2$ fraction of the times, it can attempt to select the second element at most $1/2+\epsilon/2$ fraction of the  times. 

For random arrival order, consider the feasible solution $\x$ with $x_i = 1/n$  for every $i \in \univ$. We show that no online RCRS algorithm can guarantee  each element is selected w.p. greater than $\frac{(1-1/e)}{n}$. This is because for the product distribution, w.p. $1/e$ none of the $n$ elements is active (more precisely, w.p. $(1-1/n)^n$). Hence the RCRS algorithm, which only selects active elements, selects some element w.p. $1-1/e$. This implies  on average it cannot pick every element w.p. greater than $\frac{(1-1/e)}{n}$. This example, originally shown in~\cite{CVZ-SICOMP14}, also proves that  offline CRS cannot better than $(1-1/e)$-selectable.


\paragraph*{Our techniques}
We first see the difficulty in extending Alaei's greedy approach from  a rank~$1$ matroid to a general matroid. Consider the  graphic matroid for the \emph{Hat} example (see Figure~\ref{fig:HatExample}). Suppose the base edge $(u_1,u_2)$ appears in the end of an adversarial order. Notice that any algorithm which ignores the structure of the matroid is very likely to select some pair of edges $(u_1,v_i)$ and $(v_i,u_2)$ for some $i$. Since this pair spans the base edge $(u_1,u_2)$, such an OCRS algorithm will not satisfy $c$-selectability for $(u_1,u_2)$.
To overcome this, Feldman et al.~\cite{FSZ-SODA16}  decompose the matroid into ``simpler'' matroids using $\x$. However, it is not clear how to extend their approach beyond a $1/4$-OCRS.

\tikzstyle{point}=[circle, draw, fill=black!30, inner sep=0pt, minimum width=2pt]
\tikzstyle{graphnode}=[circle, draw, fill=black!30, inner sep=0pt, minimum width=8pt]
\tikzstyle{localgraphnode}=[circle, draw, fill=red!100, inner sep=0pt, minimum width=12pt]

\begin{figure}[ht]
\setlength{\abovecaptionskip}{-5pt}
\begin{center}
\begin{tikzpicture}[thin,scale=0.75]

	\foreach \y in {1,2,4,5}{
		\draw [thick] (-2,0) to (0,\y);
		\draw [thick] (0,\y) to (2,0);
	}
	\node at (0,1) [graphnode,label=left:$v_1$]{};
	\node at (0,2) [graphnode,label=left:$v_2$]{};
	\node at (0,4) [graphnode,label=left:$v_{n-1}$]{};
	\node at (0,5) [graphnode,label=left:$v_n$]{};
		
	\foreach \y in {2.75,3,3.25}{
		\node at (0,\y) [point]{};	
	}			
	\draw [thick] (-2,0) to (2,0);
	\node at (-2,0) [graphnode,label=below:$u_1$]{};
	\node at (2,0) [graphnode,label=below:$u_2$]{};	
\end{tikzpicture}
\end{center}
\caption{The Hat example on $n+2$ vertices. The following $\x$ belongs to the graphic matroid: $x_e=1/2$ for  $e=(u_i,v_j)$ where $i\in \{1,2\}$ and $j \in \{1,\ldots, n\}$, and $x_e=1$ for $e=(u_1,u_2)$.}
\label{fig:HatExample}
\end{figure}
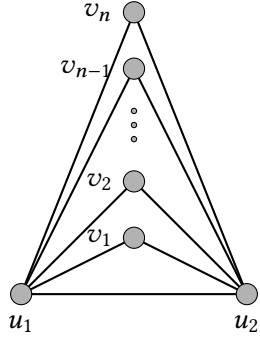

In this paper we take an alternate LP based approach to design OCRSs, which  was first used by Chekuri et al.~\cite{CVZ-SICOMP14} to design offline CRSs. The idea is to define an exponential sized linear program where each variable denotes a deterministic OCRS algorithm. The objective of this linear program is to maximize $c$ s.t. each  element is  selected at least $c$ fraction of the times ($c$-selectability). Thus to show existence of a $1/2$-OCRS, it suffices to prove this linear program has value $c \geq 1/2$. In \S\ref{sec:ExAnteToOCRS} we prove this by showing  that    the dual LP has value  at least $1/2$ because it can be interpreted as an ex-ante prophet inequality. 

Next, to show  there exists a $1/2$ approximation ex-ante prophet inequality, our approach is inspired from the matroid prophet inequality of Kleinberg and Weinberg~\cite{KW-STOC12}. They give an online algorithm that gets at least half of the expected offline optimum for the product distribution (independent r.v.s). 
Unfortunately, their techniques do not directly extend because the ex-ante relaxation objective could be significantly higher than for the product distribution (this is known as the \emph{correlation gap}, which can be $e/(e-1)$~\cite{ADSY-OR12,CCPV-SICOMP11}). 
Our primary technique  is to view the ex-ante relaxation solution as a ``special kind'' of a correlated value distribution.
Although prophet inequalities are not possible for general correlated  distributions~\cite{HillKertz-Journal92}, we show that in this special case  the original proof of the matroid prophet inequality algorithm  retains its $1/2$ approximation after some modifications.

\subsection{Further Related Work} \label{sec:related}

Krengel and Sucheston gave the first tight $1/2$-single item prophet inequality~\cite{Krengel-Journal78,Krengel-Journal77}. 
The connection between multiple-choice prophet inequalities and mechanism design was recognized in~\cite{HKS-AAAI07}; they proved a prophet inequality for uniform matroids. This bound was later improved by Alaei~\cite{Alaei-FOCS11} using the \emph{Magician's problem}, which is an  OCRS in disguise.
Chawla et al.~\cite{CHMS-STOC10} further developed the connection between prophet inequalities and mechanism design, and showed how to  be $O(1)$-prophet inequality for general matroids in a variant  where the algorithm may choose the element order. Yan~\cite{Yan-SODA11} improved this result to $e/(e-1)$-competitive using the \emph{correlation gap} for submodular functions,  first studied  in~\cite{ADSY-OR12,CCPV-SICOMP11}.
Chekuri et al.~\cite{CVZ-SICOMP14} adapted correlation gaps to a polytope to design CRSs. Improved correlation gaps were presented in~\cite{Yan-SODA11, GL-FSTTCS17}.
The {matroid prophet inequality} was first explicitly formulated in~\cite{KW-STOC12}. Feldman et al.~\cite{FSZ-SODA16} gave an alternate proof, and extended to Bernoulli submodular functions, using OCRSs. Finally, information theoretic $O(\poly\log(n))$-prophet inequalities are also known for general downward-closed constraints~\cite{Rubinstein-STOC16,RS-SODA17}.

The prophet secretary notion was first introduced in~\cite{EHLM-SIDMA17},  where the elements arrive in a uniformly random order and draw their values from known independent distributions. Their results have been recently improved~\cite{EHKS-SODA18,ACK-ArXiv17}. There is a long line of work on studying the \emph{commitment constraints} for combinatorial probing problems, e.g., see~\cite{GM-STOC07,GN-IPCO13,GNS-SODA16,GJSS-EC18}. In these models the algorithm starts with some stochastic knowledge about the input and on probing an element has to irrevocably commit if the element is to be included in the final solution. A common approach to handle such a constraint is using a prophet inequality/OCRS.


\section{OCRS Assuming an Ex-Ante Prophet Inequality} \label{sec:ExAnteToOCRS}
In this section we prove Theorem~\ref{thm:reduction}, showing how to reduce the problem of designing an OCRS to a prophet inequality where we compare ourself to the ex-ante relaxation instead of the expected offline maximum.

\subsection{Using LP Duality}

Given a finite ground set $\univ$ with $n = |\univ|$ and a {downward-closed} family of feasible  subsets $\calF \subseteq 2^\univ$, let $P_{\calF} \subseteq [0, 1]^\univ$ be the convex hull of all characteristic vectors of feasible sets. Let $\x \in P_{\calF}$ and $\Rx$ denote a random set containing each element $i\in \univ$ independently w.p. $x_i$.
For offline CRSs, let $\Phi^*$ be the set of valid offline \emph{deterministic} mappings; i.e., $\phi : 2^{\univ} \to \calF$ is in $\Phi^*$ iff $\phi(A) \subseteq A$ and $\phi(A) \in \calF$ for all $A \subseteq \univ$. 
For $\phi \in \Phi^*$ and $i \in \univ$, let $q_{i, \phi} := \Pr_{\Rx} [i \in \phi(\Rx)]$ denote the probability of selecting $i$ if the CRS executes $\phi$. 
The following LP relaxation, introduced by Chekuri et al.~\cite{CVZ-SICOMP14}, finds a $c$-selectable \emph{randomized} CRS. 
It has  variables $\{ \lambda_{\phi} \}_{\phi \in \Phi^*}$  and $c$. 
\begin{align*}
\mbox{$\max_{\bm{\lambda},c}$ }  & ~c && \\
\mbox{s.t. } & \sum_{\phi \in \Phi^*} q_{i, \phi} \lambda_{\phi} \geq x_i \cdot c && i \in \univ \\
& \sum_{\phi \in \Phi^*} \lambda_{\phi}  = 1  && \\
& \lambda_{\phi} \geq 0 && \forall \phi \in \Phi^* 
\end{align*}
Observe that if the above LP has value  $c$, there exists a randomized $c$-CRS. This is because we can randomly select  one of the $\phi$'s w.p. $\lambda_{\phi}$, and the constraint
$\sum_{\phi \in \Phi^*} q_{i, \phi} \lambda_{\phi} \geq x_i \cdot c $ ensures $c$-selectability for every $i\in \univ$.
Chekuri et al. noticed that  by strong duality, to prove the above LP has value at least $c$,  it suffices to show that the following dual program has value at least $c$.
It has variables $\{ y_i \}_{i \in \univ}$ and $\mu$. 
\begin{align*}
\mbox{$\min_{\bm{y},\mu}$}  & ~\mu && \\
\mbox{s.t. } & \sum_{i \in \univ} q_{i, \phi} y_i \leq \mu  && \phi \in \Phi^* \\
& \sum_{i \in \univ} x_i y_i  = 1  && \\
& y_i \geq 0 && \forall i \in \univ 
\end{align*}

To design OCRSs (RCRSs), we take a similar approach as Chekuri et al and let $\Phi^*$ be the set of all {\em deterministic online algorithms}. Formally,
$\phi : 2^{\univ} \times 2^{\univ} \times \univ \to \{ 0, 1 \}$ belongs to $\Phi^*$ 
iff $\phi(A, B, i) = 1$ only for $B \subseteq A$, $i \not\in A$, and $B \cup \{ i \} \in \calF$. 
Intuitively, $\phi(A, B, i) = 1$ indicates that the online algorithm selects element $i$ in the current iteration after processing elements in $A$ and selecting elements in $B$. 
Let $q_{i, \phi}$ denote the probability of selecting $i$ if the OCRS (RCRS) executes $\phi$, where for RCRS we also take probability over the random order.
By the above duality argument, to show existence of a $c$-OCRS ($c$-RCRS) it suffices to prove the dual LP has value at least $c$. We prove this by showing that for any $\bm{y} \geq 0$ s.t. $\sum_{i \in \univ} x_i y_i  = 1$, there exists $\phi \in \Phi^*$ such that $\sum_{i \in \univ} q_{i, \phi} y_i \geq c$. 

Consider a Bernoulli prophet inequality instance  where each element $i \in \univ$ has value $y_i$ with probability $x_i$, and $0$ otherwise. Since $\x \in P_{\calF}$, notice that $\sum_{i \in \univ} x_i y_i=1$ is exactly the value of the ex-ante relaxation~\eqref{eq:exanteUpper} for this instance. Thus, a $c$-approximation ex-ante prophet inequality implies there exists a $\phi \in \Phi^*$ with value at least $c$. By linearity of expectation, the value of $\phi$ is $\sum_{i \in \univ} q_{i, \phi} y_i$, which proves $\sum_{i \in \univ} q_{i, \phi} y_i \geq c$.

\subsection{Solving the LP Efficiently}
While the original primal LP has an exponential number of variables, we can compute an OCRS (or RCRS) that achieves  value at least $c$  as follows. 
In the dual program, given $\bm{y}$ s.t. $\sum_i x_i y_i = 1$, we can use the ex-ante prophet inequality to find $\phi \in \Phi^*$ with value $\sum_i q_{i, \phi} y_i \geq c$ in polynomial time. (Notice  $q_{i, \phi}$ can be computed in polynomial time because  the adversarial order is known to the OCRS algorithm.) 
This implies  for any $\epsilon > 0$, the polytope $Q_{c - \epsilon} := \{ \bm{y} : \bm{y} \geq 0, \sum_{i} x_i y_i = 1, \sum_{i} q_{i,\phi}{y_i} \leq c - \epsilon \mbox{ for all } \phi \in \Phi^* \}$ is empty. 

Since we have an efficient {\em separation oracle} (for any $y$, we can find a violated constraint in polynomial time) for $Q_{c - \epsilon}$, by running the ellipsoid algorithm~\cite{GLS-Book81} we can find a subset $\Phi' \subseteq \Phi^*$ with $|\Phi'| = \poly(n)$
in polynomial time such that $Q'_{c - \epsilon} := \{ \bm{y} : \bm{y} \geq 0, \sum_{i} x_i y_i = 1, \sum_{i} q_{i,\phi}{y_i} \leq c - \epsilon \mbox{ for all } \phi \in \Phi' \}$ is empty. Now the following  linear program, which has a polynomial number of variables and constraints, with optimal value at least $c - \epsilon$ can be solved efficiently. 
\begin{align*}
\mbox{$\max_{\bm{\lambda},c}$ }  & ~c && \\
\mbox{s.t. } & \sum_{\phi \in \Phi'} q_{i, \phi} \lambda_{\phi} \geq x_i \cdot c && i \in \univ \\
& \sum_{\phi \in \Phi'} \lambda_{\phi}  = 1  && \\
& \lambda_{\phi} \geq 0 && \forall \phi \in \Phi' 
\end{align*}

\section{Ex-Ante Prophet Inequalities for a Matroid}
\label{sec:prophet}
This section proves Theorem~\ref{thm:ExAnteProphetIneq} by designing for a matroid 
a $1/2$-ex-ante prophet inequality  under  adversarial arrival and 
a $(1 - 1/e)$-ex-ante prophet inequality  under  random arrival.

\subsection{Notation} \label{sec:OCRSNotation}

Let $\bval \sim \calD$ be a set of random element values $\{\val_1, \dots, \val_n\}$ where each $\val_i$ is independently drawn from $\calD_i$. Let $\x$ be the optimal solution to the ex-ante relaxation in \eqref{eq:exanteUpper} for a given matroid $\M = (\univ, \mathcal{I})$.
For $i\in \univ$,  denote
\begin{align} \label{eq:DefnYiOCRS}
y_i := \E_{\val_i\sim \calD_i}[\val_i \mid \text{$\val_i$ takes value in its top $x_i$ quantile}]. 
\end{align}
Since $\x \in P_{\calm}$, we can write it as a convex combination of independent sets in the matroid. In particular, this gives a correlated distribution $\hcalD$ over independent sets of $\M$ such that 
for each $i\in \univ$, we have $\Pr_{I \sim \hcalD}[i \in I] = x_i$. 
Let $\hbval = \{ \hval_1, \dots, \hval_n \}$ be a set of random values obtained by  sampling $I\sim \hcalD$ and setting $\hval_i = y_i$ for $i \in I$, and $\hval_i = 0$ otherwise. 
Notice  the optimal value of \eqref{eq:exanteUpper} is $\sum_i x_i y_i$ and for each $i\in \univ$, we have $ \E[\hval_i] = x_i y_i$. 

We need the following notation to describe our algorithms.
\begin{defn} \label{def:RemCost}
For any vector $\hbval$ denoting values of elements of  $\univ$ and any $A\subseteq \univ$, we define:
\begin{itemize}
\item Let  $\Opt(\hbval \mid A) {\subseteq \univ \setminus A}$ denote the maximum value independent set  in the contracted matroid $\M/A$.

\item Let  $R(A,\hbval):= \sum_{i\in \Opt(\hbval \mid A)} \hat{\val}_i $ denote the \emph{remaining value} after selecting set $A$.
\end{itemize}
\end{defn}

We next define a \emph{base price} of for every element $i$.
\begin{defn} For $A \in \mathcal{I}$ denoting an independent set of  elements accepted by our algorithm, we define
\begin{itemize}
\item Let $\baseprice_i(A,\hbval) := R(A,\hbval) - R(A\cup \{ i \},\hbval)$  denote a threshold for element $i$.
\item Let $\baseprice_i(A) := \E_{\hbval \sim \hcalD}[\baseprice_i(A,\hbval)]$ denote  the \emph{base price} for element $i$.
\end{itemize}

\end{defn}

\subsection{Reducing to Bernoulli Distributions} \label{sec:ReducToBernoulliOCRS}
In this section we show that it suffices to only prove Theorem~\ref{thm:ExAnteProphetIneq} for Bernoulli distributions. 

\begin{lemma} \label{lem:OCRSGenToBern}
If there exists an $\alpha$-approximation ex-ante prophet inequality for Bernoulli distributed independent random values then there exists  an  $\alpha$-approximation ex-ante prophet inequality for general distributed independent random values.
\end{lemma}
\begin{proof}
Given a prophet inequality instance where $\bval \sim \calD$ for a general distribution $\calD$, consider a new Bernoulli prophet inequality instance $\bval' \sim \calD'$ where for each $i \in \univ$, r.v. $\val'_i \sim \calD_i'$ independently takes value $y_i$ (defined in \eqref{eq:DefnYiOCRS}) w.p. $x_i$, and is $0$ otherwise. Since the optimal ex-ante fractional value for both the general and Bernoulli instance is the same, to prove this theorem we use an  ex-ante prophet inequality for the Bernoulli instance to design an ex-ante  prophet inequality  for the general instance with the same expected value. 

On arrival of an element $i$, consider an algorithm for the general distribution that treats $i$ is active  iff  $\val_i$ takes value in its top $x_i$ quantile. If active, the algorithm asks the ex-ante  prophet inequality of the Bernoulli instance to decide whether to select $i$. We claim that the expected value of this algorithm is $\alpha \cdot \sum_i x_i y_i$, which will prove this theorem.
The claim is true because for the above algorithm  each element $i$ is active independently w.p. exactly $x_i$, and conditioned on being active its expected value is exactly $y_i$. Thus by linearity of expectation, the expected value is the same as the Bernoulli instance, which is $\alpha \cdot \sum_i x_i y_i$.
\end{proof}

\subsection{Adversarial Order} \label{sec:adversMatroidOCRS}
We prove the optimal ex-ante prophet inequality for a matroid under the adversarial arrival. 

\begin{theorem} \label{thm:ExAntMatrAdvers}
For matroids, there exists a  $1/2$-approximation ex-ante prophet inequality   for  adversarial arrival order.
\end{theorem}

Given the notation and definitions in \S\ref{sec:OCRSNotation}, the proof of Theorem~\ref{thm:ExAntMatrAdvers} is  similar to the proof of the matroid prophet  inequality in~\cite{KW-STOC12}.

By Lemma~\ref{lem:OCRSGenToBern}, we know it suffices to prove this theorem only for Bernoulli distributions. Consider $\bval \sim \calD$ as the input to our online algorithm, where $\val_i$ takes value $y_i$ w.p. $x_i$ and is $0$ otherwise. Given $\bval$, our algorithm is deterministic and let $A := A(\bval)$ denote the set of elements that it selects. Relabel the elements such that the arrival  order of the elements is $1,\ldots, n$.
Let $A_i = A \cap \{ 1, \dots, i \}$. 

Our algorithm selects the next element $i$ iff both $\val_i > \T_i := \alpha \cdot \baseprice_i(A_{i-1}) $ and selecting $i$ is feasible in $\M$, where $\alpha =\frac12$. Thus, the total value of algorithm  $\Alg := \sum_{i\in A} \val_i = \Rev + \Util$, where
\[ \textstyle{ \Rev := \sum_{i\in A} \T_i \qquad \text{and} \qquad \Util := \sum_{i\in A} (v_i-\T_i)^+ .}
\]
Since $\sum_{i\in\univ} x_i y_i$  is the optimal value of \eqref{eq:exanteUpper}, 
to prove Theorem~\ref{thm:ExAntMatrAdvers} it suffices to show  $\E[\Alg] = \E[\Rev] + \E[\Util] \geq \alpha \cdot \sum_{i\in\univ} x_i y_i$. 

We keep track of the algorithm's progress using the following \emph{residual} function: 
\[ r(i) := \E_{\bval \sim \calD, \hbval \sim \hcalD}[R(A_{i-1},\hbval)]. \]  
Clearly, $r(0) = \sum_{i\in\univ} x_i y_i$. 
In the following  Lemma~\ref{lem:revAdvProph} and Lemma~\ref{lem:utilAdvProph}, we use the residual function to lower bound  $\E[\Rev]$ and $\E[\Util]$.


\begin{lemma}\label{lem:revAdvProph}
$\E_{\bval \sim \calD}[\Rev] = \alpha \cdot \Big( r(0) -  r(n) \Big).$
\end{lemma}
\begin{proof}
From the definition of $\Rev$, we get
\begin{align*}
\Rev &= \alpha \cdot \sum_{i\in A} b_i(A_{i-1}) ~ = ~ \alpha \cdot \sum_{i \in A} \Big( \E_{\hbval}[R(A_{i-1},\hbval)] - \E_{\hbval}[R(A_{i-1} \cup \{i\},\hbval)] \Big) \\
&= \alpha \cdot \sum_{i \in A} \Big( \E_{\hbval}[R(A_{i-1},\hbval)] - \E_{\hbval}[R(A_{i} ,\hbval)] \Big)
~=~  \alpha \cdot \Big( \E_{\hbval}[R(A_{0},\hbval)] - \E_{\hbval}[R(A,\hbval)]  \Big).
 \end{align*}
 Taking expectation over $\bval \sim \calD$ and using  definitions of $r(0)$ and $r(n)$, the lemma follows.
\end{proof}

\begin{lemma}\label{lem:utilAdvProph}
$E_{\bval \sim \calD}[\Util] \geq    (1-\alpha) \cdot r(n).$ 
\end{lemma}
\begin{proof} 
We prove the following two inequalities:
\begin{equation}
\E_{\bval \sim \calD}[\Util] 
\quad \geq \quad 
\E_{\bval \sim \calD, \hbval \sim  \hcalD} \Big[\sum_{i \in \Opt(\hbval \mid A)} (\hval_i - \T_i)^+ \Big] 
\label{eq:magic1}
\end{equation}
and 
\begin{equation}
\E_{\bval \sim \calD, \hbval \sim \hcalD} \Big[\sum_{i \in \Opt(\hbval \mid A)} (\hval_i - \T_i)^+ \Big]  \quad  \geq \quad 
(1-\alpha) \cdot \E_{\bval \sim \calD, \hbval \sim \hcalD} [R(A,\hbval)].
\label{eq:magic2}
\end{equation}
Lemma~\ref{lem:utilAdvProph} now follows by summing  \eqref{eq:magic1} and \eqref{eq:magic2}, and using $r(n)= \E_{\bval \sim \calD, \hbval \sim \hcalD} [R(A,\hbval)]$.

To prove \eqref{eq:magic1}, notice that  for any  $i$ not selected by the algorithm,  $\val_i \leq \T_i$. This implies
\[
\E_{\bval \sim \calD}[\Util]  = \E_{\bval} \Big[\sum_{i \in A} (\val_i - \T_i)^+\Big] 
= 
\E_{\bval} \Big[\sum_{i \in \univ} (\val_i - \T_i)^+ \Big] .
\]
Now observe that for any fixed $i$ and $\val_1, \dots, \val_{i - 1}$, the threshold $\T_i$ is determined. Since $\val_i$ and $\hval_i$ are independent random variables with the same distribution, we get 
\[
\E_\bval[(\val_i - \T_i)^+ | \val_1, \dots, \val_{i - 1}] = 
\E_{\bval, \hbval} [(\hval_i - \T_i)^+ | \val_1, \dots, \val_{i - 1}]. 
\]
This implies
\begin{align*}
\E_{\bval \sim \calD}[\Util] = \E_{\bval} \Big[\sum_{i \in \univ} (\val_i - \T_i)^+ \Big] 
 = \E_{\bval, \hbval} \Big[\sum_{i \in \univ} (\hval_i - \T_i)^+ \Big] 
 \geq 
\E_{\bval, \hbval} \Big[\sum_{i \in \Opt(\hbval \mid A)} (\hval_i - \T_i)^+ \Big].
\end{align*}

Finally, to prove  \eqref{eq:magic2}, we have
\begin{align*}
\E_{\bval, \hbval} [R(A,\hbval)] = \E_{\bval, \hbval} \Big[\sum_{i \in \Opt(\hbval \mid A) } \hval_i \Big]
& \leq \E_{\bval, \hbval} \Big[\sum_{i \in \Opt(\hbval \mid A) } \T_i \Big] + \E_{\bval, \hbval} \Big[\sum_{i \in \Opt(\hbval \mid A)} (\hval_i - \T_i)^+ \Big] \\
& \leq \alpha \cdot \E_{\bval, \hbval} [R(A,\hbval)]  + \E_{\bval, \hbval} \Big[\sum_{i \in \Opt(\hbval \mid A) } (\hval_i - \T_i)^+ \Big],
\end{align*}
where the first inequality uses $\hval_i \leq \T_i + (\hval_i - \T_i)^+$  
and the second inequality uses Claim~\ref{lem:KWProp2Advers} for $S = \Opt(\hbval \mid A)$.  After rearranging, this implies \eqref{eq:magic2}.
\end{proof}
We  need the following Claim~\ref{lem:KWProp2Advers} in the proof of Lemma~\ref{lem:utilAdvProph}.
\begin{claim} \label{lem:KWProp2Advers}
For every pair of disjoint sets $A, S$ such that $A \cup S \in \M$, 
\begin{equation}
\alpha \cdot \E_{\hbval \sim \hcalD} \Big[\sum_{i \in S} R(A_{i - 1},\hbval) - R(A_{i - 1} \cup \{ i \},\hbval) \Big] 
= \sum_{i \in S} \T_i \leq \alpha \cdot  \E_{\hbval \sim \hcalD} [R(A,\hbval)]. 
\label{eq:balance2}
\end{equation}
\end{claim}
\begin{proof}
This  directly follows from \cite{KW-STOC12}, as they proved it for every fixed $\hbval$. The proof is similar to Claim~\ref{lem:KWProp2} in the next section.
\end{proof}

\begin{proof}[Proof of Theorem~\ref{thm:ExAntMatrAdvers}]
Using Lemma~\ref{lem:revAdvProph} and Lemma~\ref{lem:utilAdvProph}, and substituting $\alpha=\frac12$, we get
\[	\E[\Alg] = \E[\Util] + \E[\Rev] \geq \frac12 \cdot r(0) = \frac12 \cdot  \sum_{i\in\univ} x_i y_i. \qedhere
\]
\end{proof}


\subsection{Random Order}

 We prove the optimal ex-ante prophet inequality for a matroid for random arrival. 
\begin{theorem} \label{thm:ExAnteProphSecMatr}
For matroids, there exists  a $(1-1/e)$-approximation ex-ante prophet inequality   for uniformly random arrival order.
\end{theorem}


The proof of Theorem~\ref{thm:ExAnteProphSecMatr} is  similar to the matroid prophet secretary inequality in~\cite{EHKS-SODA18}.
We consider the model where each item chooses the arrival time from $[0, 1]$ uniformly and independently, which is equivalent to the random permutation model. 
{Starting with $A_0=\emptyset$, let $A_t$ denote the set of accepted elements by our algorithm \emph{before} time $t$. This is a random variable that depends on the values $\bval$ and arrival times $\barrivaltime$.}
For $t\in [0,1]$, let 
\[\alpha(t) := 1-\exp(t-1) . \]
Suppose an  element $i$ arrives at time $t$, then our algorithm selects $i$ iff both $\val_i > \alpha(t) \cdot \baseprice_i(A_t) $ and selecting $i$ is feasible in $\M$.

\IGNORE{
Thus, the total value of algorithm  $\Alg := \sum_{i\in A} \val_i = \Rev + \Util$, where
\[\Rev := \sum_{i\in A} T_i \qquad \text{and} \qquad \Util := \sum_{i\in A} (v_i-T_i)^+ .
\]
Since $\sum_{i\in\univ} x_i y_i$  is the optimal value of \eqref{eq:exanteUpper}, 
to prove Theorem~\ref{thm:ExAntMatrAdvers} it suffices to show  $\E[\Alg] = \E[\Rev] + \E[\Util] \geq \alpha \cdot \sum_{i\in\univ} x_i y_i$. 
}

Similar to \S\ref{sec:adversMatroidOCRS}, we keep track of the algorithm's progress using the  \emph{residual} function
\[ r(t) := \E_{\bval \sim \calD, \hbval \sim \hcalD,\barrivaltime}[R(A_t,\hbval)],\] 
where $A_t$ is a function of $\bval$ {and $\barrivaltime$}. Clearly, $r(0) = \sum_{i\in \univ} x_i y_i$. 

\begin{claim}\label{claim:revMat}
$\E_{\bval  \sim \calD,\barrivaltime}[\Rev] = - \displaystyle{\int_{t=0}^1} \alpha(t) \cdot r'(t) dt.$
\end{claim}
\begin{proof}
This follows directly from the definition of $\Rev$. See~\cite{EHKS-SODA18} for  details.
\end{proof}

\IGNORE{
\begin{proof}
Consider the time from $t$ to $t + \epsilon$ for some $t \in [0,1]$, $\epsilon > 0$. Let us fix the arrival times $\barrivaltime$ and values $\bval$ of all elements. This also fixes the sets $(A_t)_{t \in [0,1]}$. Let $i_1, \ldots, i_k$ be the arrivals between $t$ and $t+\epsilon$ that get accepted in this order. Note that it is also possible that $k=0$. The revenue obtained between $t$ and $t+\epsilon$ is now given as
\begin{align*}
\Rev_{\leq t+\epsilon} - \Rev_{\leq t} & = \sum_{j=1}^k \alpha(t_{i_j}) \baseprice_{i_j}(A_{t_{i_j}})\\
& = \sum_{j=1}^k \alpha(t_{i_j}) \E_{\hbval} \left[ R(A_t \cup \{i_1, \ldots i_{j-1}\}, \hbval) - R(A_t \cup \{i_1, \ldots i_j\}, \hbval) \right] \\
& \geq \alpha(t+\epsilon) \E_{\hbval} \left[ R(A_t, \hbval) - R(A_{t+\epsilon}, \hbval) \right].
\end{align*}
Taking the expectation over $\bval$ and $\barrivaltime$, we get by linearity of expectation
\[
\E_{\bval,\barrivaltime}[\Rev_{\leq t+\epsilon}] - \E_{\bval,\barrivaltime}[\Rev_{\leq t}] \geq \alpha(t+\epsilon) (r(t) - r(t+\epsilon)).
\]
By the same argument, we also have
\[
\E_{\bval,\barrivaltime}[\Rev_{\leq t+\epsilon}] - \E_{\bval,\barrivaltime}[\Rev_{\leq t}] \leq \alpha(t) (r(t) - r(t+\epsilon)) .
\]
In combination, we get that
\[
\frac{d}{dt} \E_{\bval,\barrivaltime}[\Rev_{\leq t}] = - \alpha(t) r'(t) ,
\]
which implies the claim.
\end{proof}
}

\begin{lemma}\label{lem:utilMat}
$\E_{\bval \sim \calD,\barrivaltime}[\Util] \geq  \displaystyle{\int_{t=0}^1}  (1-\alpha(t)) \cdot r(t) dt.$
\end{lemma}
\begin{proof} The utility for element $i$ arriving at time $t$ is given by 
\begin{align*} \E_{\bval,\barrivaltime}[u_i \mid \arrivaltime_i = t] &= \E_{\bval,\barrivaltime_{-i}} \left[ \left(\val_i - \alpha(t) \cdot \baseprice_i(A_t) \right) ^+ \cdot \one_{i \not \in \Span(A_t)} \growingmid \arrivaltime_i = t \right] .
\end{align*}
Observe that $A_t$ does not depend on $\val_i$ if $\arrivaltime_i = t$ because it includes only the acceptances \emph{before} $t$. It does not depend on $\hat{\val_i}$ either, as $\hat{\val_i}$ is only used for analysis purposes and not known to the algorithm. Since $\val_i$ and $\hat{\val_i}$ are identically distributed, we can also write 
\begin{align}\label{eq:AtXtIndependence}
\E_{\bval \sim \calD,\barrivaltime}[u_i \mid \arrivaltime_i = t] = \E_{\bval \sim \calD,\hbval \sim \hcalD,\barrivaltime_{-i}} \left[ \left(\hval_i - \alpha(t) \cdot \baseprice_i(A_t) \right) ^+ \cdot \one_{i \not \in \Span(A_t)} \growingmid \arrivaltime_i = t \right].
\end{align}
Now observe that element $i$ can belong to $\Opt(\hbval\mid A_t)$ only if it's not already in $\Span(A_t)$, which implies  $\one_{i \not \in \Span(A_t)} \geq \one_{i \in \Opt(\hbval\mid A_t)}$. Using this and removing non-negativity, we get
\[ \E_{\bval,\barrivaltime}[u_i \mid \arrivaltime_i = t] \geq \E_{\bval,\hbval,\barrivaltime_{-i}} \left[ \left(\hval_i - \alpha(t) \cdot \baseprice_i(A_t) \right)  \cdot \one_{i \in \Opt(\hbval\mid A_t)} \growingmid \arrivaltime_i = t \right] .
\]
Now we use Lemma~\ref{lem:withwithouiMatroid}  to remove the conditioning on element $i$ arriving at time $t$ as this gives a valid lower bound on expected utility, 
\begin{align} \label{eq:utilBuyerIMatroid}
 \E_{\bval,\barrivaltime}[u_i \mid \arrivaltime_i = t] \geq \E_{\bval,\hbval,\barrivaltime} \left[\left(\hval_i - \alpha(t) \cdot \baseprice_i(A_t) \right)  \cdot \one_{i\in \Opt(\hbval\mid A_t)} \right].
\end{align}

We can now lower bound sum of all the utilities using  Eq.~\eqref{eq:utilBuyerIMatroid} to get
\begin{align*}
\E_{\bval,\barrivaltime}[\Util] & = \sum_i \int_{t=0}^1 \E_{\bval,\barrivaltime}[u_i \mid \arrivaltime_i = t] \cdot  dt \\
&\geq \sum_i \int_{t=0}^1 \E_{\bval,\hbval \sim \hcalD,\barrivaltime} \left[\left(\hval_i - \alpha(t) \cdot \baseprice_i(A_t) \right)  \cdot \one_{i\in \Opt(\hbval\mid A_t)} \right] \cdot dt.
\end{align*}
By moving the sum over elements inside the integrals, we get
\begin{align*}
\E_{\bval,\barrivaltime}[\Util]  &\geq  \int_{t=0}^1  \E_{\bval,\hbval,\barrivaltime} \Big[ \sum_i \left(\hval_i - \alpha(t) \cdot \baseprice_i(A_t) \right)  \cdot \one_{i\in \Opt(\hbval\mid A_t)} \Big] \cdot dt\\
& =\int_{t=0}^1  \E_{\bval,\hbval,\barrivaltime} \Big[   R(A_t,\hbval) -  \alpha(t) \cdot \sum_{i\in \Opt(\hbval\mid A_t)}  \baseprice_i(A_t)   \Big] \cdot dt .
\end{align*}
Finally, using Claim~\ref{lem:KWProp2} for $S=\Opt(\hbval\mid A_t)$, we get 
\[ \E_{\bval,\barrivaltime}[\Util]  \geq \int_{t=0}^1  \E_{\bval,\hbval,\barrivaltime} \left[ \left(1 - \alpha(t) \right) \cdot R(A_t,\hbval)  \right] \cdot dt. \qedhere
\]
\end{proof}


\begin{proof}[Proof of Theorem~\ref{thm:ExAnteProphSecMatr}]
Using   Lemma~\ref{lem:utilMat} and Claim~\ref{claim:revMat}, we get
\begin{align*}
	\E[\Alg] &= \E[\Rev] + \E[\Util]  \\
	& \geq  - \int_{t=0}^1 \alpha(t) \cdot r'(t) \cdot dt + \int_{t=0}^1  (1-\alpha(t)) \cdot r(t) \cdot dt \\
	&= \int_{t=0}^1 r(t) \cdot (1 - \alpha(t) + \alpha'(t))\cdot dt - [r(t) \cdot \alpha(t)]_{t=0}^{1}.
\end{align*}
Notice that for $\alpha(t) = 1-e^{t-1}$, we have $1 - \alpha(t) + \alpha'(t) = 0$. Hence, we get
\[		\E[\Alg] \geq - [r(t) \cdot \alpha(t)]_{t=0}^{1} = \left( 1- \frac1e \right) \cdot r(0) = \left( 1- \frac1e \right) \cdot \sum_{i \in \univ} x_i y_i. \qedhere
\]
\end{proof}

Finally, we prove the missing Lemma~\ref{lem:withwithouiMatroid} that removes the conditioning on  $i$ arriving at $t$.
\begin{lemma}\label{lem:withwithouiMatroid} For any $i$, any time $t$, and any fixed $\bval,\hbval$,  we have
\begin{align*} 
\E_{\barrivaltime_{-i}} \left[ \left(\hval_i - \alpha(t) \cdot \baseprice_i(A_t) \right) \cdot \one_{i\in \Opt(\hbval\mid A_t)} \mid \arrivaltime_i = t \right]  ~\geq~ \E_{\barrivaltime} \left[\left(\hval_i - \alpha(t) \cdot \baseprice_i(A_t) \right)  \cdot \one_{i\in \Opt(\hbval\mid A_t)} \right].
\end{align*}
\end{lemma}
\begin{proof}
We prove the lemma for any fixed $\barrivaltime_{-i}$. Suppose we draw a uniformly random $\arrivaltime_i\in [0,1]$. Observe that  if $\arrivaltime_i \geq t$ then we have equality in the above equation because set $A_t$  is the same both with and without $i$. This is also the case when $\arrivaltime_i<t$ but $i$ is not selected into $A_t$. Finally, when $\arrivaltime_i<t$ and $i \in A_t$  we have $\one_{i\in \Opt(\hbval\mid A_t)} =0$ in the presence of element $i$ (i.e., RHS of lemma), making the inequality trivially true.
\end{proof}

\begin{claim}\label{lem:KWProp2} For any fixed $\bval,\barrivaltime$, time $t$, and set of elements $S \subseteq \univ$ that is independent in the matroid $\M/A_t$, we have
\[
\sum_{i\in S}  \baseprice_i(A_t)  \leq  \E_{\hbval} \left[ R(A_t,\hbval)\right] .
\]
\end{claim}

\begin{proof}
By definition
\[
\textstyle{\sum_{i\in S}  \baseprice_i(A_t) = \E_{\hbval} \left[\sum_{i\in S} \left( R(A_t ,\hbval)-R(A_t \cup \{ i \},\hbval)  \right) \right]. }
\]
Fix the values $\hbval$ arbitrarily, we also have
\[
\sum_{i\in S} \left( R(A_t ,\hbval)-R(A_t \cup \{ i \},\hbval) \right) \leq R(A_t,\hbval).
\]
This follows from the fact that $R(A_t ,\hbval)-R(A_t \cup \{ i \},\hbval)$ are the respective critical values of the greedy algorithm on $\M/A_t$ with values $\hbval$. Therefore, the bound follows from Lemma~3.2 in~\cite{LucierBorodin-SODA10}. An alternative proof is given as Proposition~2 in \cite{KW-STOC12} while in our case the first inequality can be skipped and the remaining steps can be followed replacing $A$ by $A_t$.

Taking the expectation over $\hbval$, the claim follows.
\end{proof}

\appendix
\section{Illustrative Examples}
\subsection{A $1/2$-OCRS for Rank $1$ Matroids} \label{sec:OCRSProofsSingleItem}

Given $\x \in [0,1]^n$ satisfying $\sum_i x_i \leq 1$, in this section we present the proof of $1/2$-selectable OCRS due to Alaei~\cite{Alaei-SICOMP14} for completeness. (He called it the Magician's problem.)

\begin{theorem}[Alaei~\cite{Alaei-SICOMP14}]
There exists a $1/2$-OCRS for a rank $1$ matroid.
\end{theorem}
Since the algorithm  selects at most $1$ element, the only decision it  makes is whether to accept the next element $i$ if it is active. The main idea is to ignore (i.e., not consider) $i$ with  maximum probability, while satisfying that on average it is considered at least $\alpha$ fraction of times for some fixed $\alpha$  (we later set $\alpha=1/2$). Thus on reaching $i$, the algorithm selects $i$ iff it is both considered and is active.
\begin{proof}
Relabel the elements s.t. the arrival order is $1,2,\ldots, n$.
For $i\in [n]$, let $r_i$ denote the probability that the algorithm \emph{reaches} $i$, i.e., it has not selected any of the elements $[n-1]$. Thus, $r_1=1$ and we want $r_n = \alpha$. Let $q_i$ denote the probability that $i$ is considered, conditioned on the event that  algorithm reaches $i$. The algorithm sets $q_{i}$ s.t.  it considers each element w.p. $\alpha$, i.e.,
\begin{align} \label{eq:setQi}
 r_i \cdot q_i =\alpha.
\end{align}
Since on reaching $i$ the algorithm accepts it only when it is both considered and  active,  
\[	r_{i+1} = r_i \cdot (1 - q_i x_i).
\]
Now using \eqref{eq:setQi}, this gives
\[ r_{i+1} = r_i  - \alpha x_i.
\]
Summing over all $i$ and using $r_1=1$, we get 
\[ \textstyle{ r_n  = r_1 - \alpha \sum_i x_i \geq  1 - \alpha. } \]
Finally, using $r_n=\alpha$, we get $\alpha \geq 1/2$.
\end{proof}

\subsection{A $(1-1/e)$-RCRS for Rank $1$ Matroids} \label{sec:RCRSProofsSingleItem}

Given $\x \in [0,1]^n$ satisfying $\sum_i x_i \leq 1$, in this section we give a simple $(1-1/e)$-selectable RCRS. As a corollary, this gives an alternate proof of the $(1-1/e)$-prophet secretary for single item due to Esfandiari et al.~\cite{EHLM-SIDMA17}.

We first notice that the random order can be emulated  by assuming each element $i$ selects a random time $t_i$ to arrive uniformly at random in the interval $[0,1]$. 
\begin{theorem}
An algorithm that selects an active  element $i$ arriving at time $t \in [0,1]$ with probability $\exp(-t \cdot x_i)$ (and ignores $i$ otherwise)  is $(1-1/e)$-selectable for a rank $1$ matroid; that is, on average this algorithm \emph{considers} (not ignore)  any element $i$ at least $(1-1/e)$ fraction of the times. 
\end{theorem}
\begin{proof} By reaching time $t$ (element $j$), let us denote the event that no element is selected before time $t$ (element $j$'s arrival). We start by noticing that for any element $i \in [n]$,
\begin{align}
 \Pr[\text{$i$  is considered}]  &= \int_{t=0}^1  \Pr[\text{$i$ is considered at time $t$ $\mid$ reach time $t$ \& $i$ arrives at $t$ }]  \notag \\
&\qquad \qquad \cdot~ \Pr[\text{reach time $t $ $\mid$  $i$ arrives at $t$}] \cdot dt \notag \\
& = \int_{t=0}^1  \exp(-t\cdot x_i) \cdot \Pr[\text{reach time $t $ $\mid$  $i$ arrives at $t$}] \cdot dt. \label{eq:iIsConsidered}
\end{align}
Now we can simplify 
\begin{align*}
 &\Pr[\text{reach time $t $ $\mid$  $i$ arrives at $t$}]   \\
 &= \prod_{j\neq i} \Big( 1 - \Pr[\text{$j$ arrives before $t$ \& is active \& is considered} \mid \text{reach $j$}] \Big) \\
 & = \prod_{j\neq i} \Big( 1- x_j \cdot \Pr[\text{$j$ arrives before  $t$ \& is considered$\mid$ reach  $j$}]  \Big)\\
 & = \prod_{j\neq i} \Big( 1- x_j \cdot \int_{a=0}^t \exp(-a \cdot x_j) \cdot da    \Big) \quad = \quad \prod_{j\neq i}  \exp(-t \cdot x_j).
\end{align*}
Now combining this equation with \eqref{eq:iIsConsidered}, we get 
\begin{align*}
 \Pr[\text{$i$  is considered}] &= \int_{t=0}^1  \exp(-t\cdot x_i) \cdot\prod_{j\neq i} \exp(-t \cdot x_j)  \cdot dt  \\
 &\geq \int_{t=0}^1 \exp(-t) \cdot dt  =  1-\frac1e, 
\end{align*}
where the  inequality uses $\sum_i x_i \leq 1$.
\end{proof}

\medskip
\noindent
{\bf Acknowledgments}.
We are thankful to Ravishankar Krishnaswamy and Deeparnab Chakrabarty for useful discussions in  early part of this project. Part of this work was done while the authors were visiting the Simons Institute for the Theory of Computing. The second author was supported in part by NSF awards CCF-1319811, CCF-1536002, and CCF-1617790.


\bibliographystyle{alpha}
\bibliography{bib}

\end{document}